\documentclass[11pt]{article}
\usepackage[utf8]{inputenc}
\usepackage[T2A]{fontenc}
\usepackage[english]{babel}
\usepackage{amsmath,amssymb,amsthm}
\usepackage[x11names, rgb]{xcolor}
\usepackage{fullpage}
\usepackage{verbatim}
\usepackage{tikz}
\usepackage{hyperref}
\usepackage[textwidth=23mm,colorinlistoftodos]{todonotes}
\usetikzlibrary{calc,arrows,shapes,backgrounds,patterns,fit,decorations,decorations.pathmorphing}

\tikzstyle{anypath}=[->,dashed]
\tikzstyle{vertex}=[draw,ellipse,inner sep=0.5mm]
\tikzstyle{inputvertex}=[draw,rectangle,inner sep=1mm]

\DeclareMathOperator{\poly}{poly}

\newtheorem{theorem}{Theorem}

\newtheorem{lemma}{Lemma}
\newtheorem{remark}{Remark}
\newtheorem{definition}{Definition}
\newtheorem{corollary}{Corollary}

\title{Families with infants: a general approach to solve hard partition problems\thanks{A revised version of this paper is available at \url{http://arxiv.org/abs/1410.2209}}}
\author{
Alexander~Golovnev
\thanks{New York University}
\and
Alexander~S.~Kulikov 
\thanks{St.~Petersburg Department of Steklov Institute of Mathematics of the Russian Academy of Sciences}
\and
Ivan~Mihajlin
\thanks{St.~Petersburg Academic University}
}
\date{}
\begin{document}

\maketitle


\begin{abstract}
We introduce a general approach for solving partition problems where the goal is
to represent a given set as a union (either disjoint or not) of subsets satisfying certain properties. Many NP-hard problems can be naturally stated as such partition problems. We show that if one can find a large enough system of so-called families with infants for a given problem, then this problem can be solved faster than by a straightforward algorithm. 
We use this approach to improve known bounds for several NP-hard problems
as well as to simplify the proofs of several known results.

For the chromatic number problem we present an algorithm with $O^*((2-\varepsilon(d))^n)$ time and exponential space for graphs of average degree~$d$. This improves the algorithm by Bj\"{o}rklund et al. [Theory Comput. Syst. 2010] that works for graphs of bounded maximum (as opposed to average) degree
and closes an open problem stated by Cygan and Pilipczuk [ICALP 2013].

For the traveling salesman problem we give an algorithm working in 
$O^*((2-\varepsilon(d))^n)$ time and polynomial space for
 graphs of average degree~$d$. The previously known results of this kind is a polyspace algorithm by Bj\"{o}rklund et al. [ICALP 2008] for graphs of bounded maximum degree and an exponential space algorithm for bounded average degree by Cygan and Pilipczuk [ICALP 2013].

For counting perfect matching in graphs of average degree~$d$ we present an algorithm with running time $O^*((2-\varepsilon(d))^{n/2})$ and polynomial space. 
Recent algorithms of this kind due to 
Cygan, Pilipczuk [ICALP 2013] and Izumi, Wadayama [FOCS 2012] (for bipartite graphs only) 
use exponential space.
\end{abstract}





\section{Introduction}

In this paper we consider algorithms for three classical hard problems: the traveling salesman problem, the chromatic number problem, and the problem of counting
perfect matchings.
$O^*(2^n)$ algorithms for the traveling salesman problem by Bellman~\cite{B1962} and Held and Karp~\cite{HK1971} are known for more than 50 years already ($n$~is the number of vertices of an input graph). The current record upper bound $O^*(2^n)$ for the chromatic number problem
is proved by Bj\"{o}rklund et al.~\cite{BHK2009} in~2006. The number of perfect matchings in an $n$-vertex graph can be computed in time $O^*(2^{n/2})$ as shown by Bj\"{o}rklund~\cite{B2012} in 2012 (this matches the bound of Ryser's algorithm~\cite{R1963} for bipartite graphs).

These problems (and many others) can be seen as partition problems:
in the chromatic number problem the goal is to partition the vertices into
independent sets; in the traveling salesman problem one needs to partition
the vertices into several Hamiltonian sets with minimal total weight; in counting perfect matchings the goal is to partition the vertices into adjacent pairs.

We show a general way to solve such partition problems in $O^*(2^n)$ time
where $n$ is the size of a set to be partitioned. The method is in some sense a rephrasing of the method from~\cite{BHK2009} where instead of the inclusion-exclusion principle we use
the fast Fourier transform (FFT). See, e.g.,~\cite{DPV2006} for an introduction to FFT. In particular, the method immediately implies an $O^*(2^n)$
upper bound for such partition problems as chromatic number, maximum $k$-cut, domatic number, bin packing.

For all three problems mentioned above (chromatic number, traveling salesman, counting perfect matchings),
improving 
the known bound for the general case
is a major open problem in the field of algorithms for NP-hard problems. However in recent years it was shown that the bound can be improved for various special cases.
In~\cite{BHKK2008, BHKK2010, CP2013} 
better upper bounds are proved for graphs of bounded degree.

We further extend our approach to get bounds of the form
$O^*((2-\epsilon)^n)$ in various special cases. Namely we show that
such a bound follows almost immediately if the corresponding partition problem possesses a certain structure. Informally, this structure can be 
described as follows. Assume that a group of people is going to an excursion and our task is to seat them into buses with several constraints
saying that a pair of people does not want to see each other in the same bus.
This is the coloring problem and it can be solved in $O^*(2^n)$ time. Assume now that we have additional constraints: the group of people
contains several infants and these infants should be accompanied by their relatives in a bus. Roughly, we prove that if the number of infants is linear
then the problem can be solved in $O^*((2-\varepsilon)^n)$ time.

Using this approach we unify several known results of this kind. An additional advantage of the approach is that it is particularly easy to use it as a black box. Namely, all one needs to do is to reveal the corresponding structure of families with infants. This way, some of the known results follow just in a few lines. By using additional combinatorial ideas we also prove the following new results.

For the chromatic number problem, Bj\"{o}rklund et al.~\cite{BHKK2010} presented an algorithm working in time
$O^*((2-\varepsilon(\Delta))^n)$ on graphs of bounded maximum degree~$\Delta=O(1)$. The algorithm is based on Yate's 
algorithm and Moebios inversion and thus uses exponential space. We extend this result to a wider class 
of bounded average degree graphs. This closes an open problem concerning the existence of such an algorithm stated
by Cygan and Pilipczuk~\cite{CP2013}.

For the traveling salesman problem on graphs of maximum degree~$\Delta=O(1)$, Bj\"{o}rklund et al.~\cite{BHKK2008} presented an algorithm working in time $O^*((2-\varepsilon(\Delta))^n)$
and exponential space. Cygan and Pilipczuk~\cite{CP2013} extended the result to graphs of bounded average (instead of maximum) degree. Both algorithms are based on dynamic programming and the savings in the running time comes from an observation that in case of bounded degree graphs an algorithm does
not need to go through all possible subsets of vertices (e.g., a disconnected subset of vertices is not Hamiltonian for sure). It is also because of the dynamic programming technique that both mentioned algorithms use exponential space. We further improve these results by showing an algorithm working in time
$O^*((2-\varepsilon(d))^n)$ and polynomial space on graphs of average degree~$d$.

Cygan and Pilipczuk~\cite{CP2013} developed an algorithm with running time $O^*((2-\varepsilon(d))^{n/2})$ and exponential space for counting perfect matching in graphs of average degree~$d$. We present an algorithm solving this problem in $O^*((2-\varepsilon(d))^{n/2})$ time and polynomial space. Several bounds of this kind are already known for bipartite graphs~\cite{ALS1991,BF2002,SW2005,RBR2009,IW2012,CP2013}.

\section{Notation}
Let $G=(V,E)$ be a simple undirected graph. Throughout the paper we implicitly assume that the set of vertices of a graph under consideration is $V=\{1,2, \dots, n\}$. For simplicity, we consider undirected graphs only (whether a graph is directed or not is only important for the traveling salesman problem; the presented algorithm works for both undirected and directed graphs).

By $d(G)$ and $\Delta(G)$ we denote the average and the maximum degree of $G$ (we omit $G$ if it is clear from the context).
$N_G(v)$ is a \emph{neighborhood} of $v$ in~$G$, i.e., all the neighbors of $v$ in~$G$ and $N_G[v]=N_G(v) \cup \{v\}$ is its \emph{closed neighborhood}.
For $S \subseteq V$, by $G[S]$ we denote a subgraph of $G$ induced by $S$. We use $G\setminus S$ as a shortcut for $G[V \setminus S]$.

We often exploit the following simple fact: one can find in $G$ an independent set of size at least $\frac{n}{\Delta(G)+1}$ in polynomial time (this is done by a straightforward greedy algorithm).

The \emph{square} of $G=(V,E)$ is a graph $G^2=(V,E')$ where $E' \supseteq E$ is
\[E'=\{(u,v) \colon \text{ there is a path of length at most $2$ from $u$ to $v$ in $G$}\} \,.\] Note that $\Delta(G^2) \le (\Delta(G))^2$ and hence one can easily find an independent set of size $\frac{n}{(\Delta(G))^2+1}$ 
in~$G^2$.

Following~\cite{CP2013}, by 
$V_{>c}$ we denote a subset of vertices $V$ of degree greater than~$c$.
$V_{<c}$, $V_{=c}$, $V_{\le c}$, $V_{\ge c}$ are defined similarly.


\section{Partition problems}\label{sec:partitions}

\begin{definition}
Let $V=\{1, \dots, n\}$, $1\le k \le n$ be an integer and $\mathcal{F}=\{\mathcal{F}_1, \ldots, \mathcal{F}_k\}$, where each $\mathcal{F}_i \subseteq 2^V$ is a family of subsets of~$V$. A~\emph{$(V,k,\mathcal{F})$-partition problem} is to represent $V$ as a disjoint union of $k$ sets from~$\mathcal{F}_i$'s:
\[V=F_1 \sqcup \ldots \sqcup F_k \text{, where } F_i \in \mathcal{F}_i, \forall 1 \le i \le k \, .\]
\end{definition}

The brute force search algorithm for this problem takes time $O^*(\max_{1 \le i \le k}|\mathcal{F}_i|^k)$.
Using FFT one can easily prove an upper bound $O^*(2^n)$ which beats the previously mentioned bound in many interesting cases.

\begin{theorem}\label{thm:simple}
A $(V,k,\mathcal{F})$-partition problem can be solved in $O^*(2^n)$ time and space.
\end{theorem}

Note that Theorem~\ref{thm:simple} immediately implies $O^*(2^n)$ upper bound for such problems as domatic number, chromatic number, maximum $k$-cut. These bounds were proved relatively recently by Bj\"{o}rklund, Husfeldt and Koivisto~\cite{BHK2009} using the inclusion-exclusion method.

Below, we formally define a combinatorial structure called families with infants that
allow to get smaller than $O^*(2^n)$ upper bounds for partition problems.

\begin{definition}
$\mathcal{R}=((R_1,r_1), \dots, (R_p,r_p))$ is a called a \emph{$(p,q)$-system of families with infants} for a $(V,k,\mathcal{F})$-partition problem if all of the following conditions are satisfied:
\begin{enumerate}
\item $pq \le n$;
\item for all $i=1,\dots,p$, $r_i \in R_i$; $r_i$ is called \emph{an infant} and all the elements of $R_i\setminus \{r_i\}$ are called \emph{relatives} of $r_i$; the sets $R_i$ are called \emph{families};
\item the size of each family $R_i$ is at most $q$;
\item all families $R_i$'s are pairwise disjoint;
\item in any valid partition each infant is accompanied by at least one of its relatives:
\begin{equation}\label{eq:mainprop}
\text{for all $1 \le i \le p, 1 \le j \le k$ and all $F \in \mathcal{F}_j$, if $r_i \in F$ then $|F \cap R_i| \ge 2$.}
\end{equation}
\end{enumerate}
\end{definition}


\begin{theorem}\label{thm:enc}
Let $\mathcal{R}=((R_1,r_1), \dots, (R_p,r_p))$ be a {$(p,q)$-system of families with infants}
for a $(V,k,\mathcal{F})$-partition problem. 
Then the problem can be solved in time and space
\begin{equation}
O^*\left(2^{n} \cdot \left(\frac{2^q-1}{2^q}\right)^p \cdot 2^q\right) \, .
\end{equation}
\end{theorem}

Roughly, the savings in the running time comes from the fact that 
while looking for a valid partition of $V$ one can avoid the case $F \cup R_i=\{r_i\}$
(i.e., instead of considering all $2^q$ possibilities of $F \cap R_i$ one considers
$2^q-1$ of them).

We always use Theorem~\ref{thm:enc} with $q=O(1)$ and $p=\Omega(n)$, which makes the running time less than $O^*(2^n)$.
\begin{corollary}\label{cor:enc}
Let $\mathcal{R}=((R_1,r_1), \dots, (R_p,r_p))$ be a {$(p,q)$-system of families with infants}
for a $(V,k,\mathcal{F})$-partition problem. If $q=O(1)$ and $p=\Omega(n)$,
then the problem can be solved in time and space $O^*((2-\varepsilon)^n)$.
\end{corollary}

As an illustration of Theorem~\ref{thm:enc} we replicate a result from~\cite{BHKK2010}. In the (decision version of) domatic number problem the question is to partition the set of vertices into $k$ dominating sets.

\begin{lemma}\label{lemma:domatic}
The domatic number problem in a graph of maximum degree $\Delta=O(1)$ can be solved in time and space $O^*((2-\varepsilon(\Delta))^n)$.
\end{lemma}

\begin{proof}\label{page:domaticproof}
The domatic number problem is a $(V,k,\mathcal{F})$-problem where each $\mathcal{F}_i$ is just the set of all dominating sets of $G$. By definition, for any $v \in V$ and any dominating set $U \subseteq V$, $N_G[v] \cap U \neq \emptyset$. This gives a straightforward construction of families with infants.

Find greedily an independent set $I \subseteq V$ of size $p=\frac{n}{\Delta^2+2}$ in $G^2$. Assume w.l.o.g. that $I=\{1, \dots, p\}$. For each $1 \le i \le p$,
let $R_i = N_G[i]$. 
At this point we have at least $n-p(\Delta+1) \ge p$ remaining vertices in $V \setminus \cup_{i=1}^{p}R_i$. So, we can extend each $R_i$ by one vertex
and declare this one additional vertex as the infant of $R_i$.

All $R_i$ have size at most $q=\Delta+2=O(1)$, the total number of $R_i$'s is $p=\frac{n}{\Delta^2+2}=\Omega(n)$. Clearly $pq \le n$. The constructed sets satisfy the property~(\ref{eq:mainprop}) by an obvious reason: each $R_i$ is a superset of $N_G[v]$ for some $v \in V$ and none of these elements is the infant of the family~$R_i$. And $U \cap N_G[v] \neq \emptyset$ for any dominating set~$U$ and any vertex~$v$, i.e.,
any dominating set always contains at least one relative of $r_i$
(even if it does not contain $r_i$).

The upper bound now follows from Theorem~\ref{thm:enc} and Corollary~\ref{cor:enc}.
\end{proof}

Another example is a faster than $O^*(2^n)$ algorithm for finding
a Hamiltonian cycle in a graph of bounded degree. This result
was given in~\cite{BHKK2008}.

\begin{lemma}\label{lemma:tspmax}
The Hamiltonian cycle problem on a graph of maximum degree $\Delta=O(1)$ can be solved in time and space $O^*((2-\varepsilon(\Delta))^n)$.
\end{lemma}

\begin{proof}\label{page:tspmax}
Guess three vertices $v_0,v_1,v_2 \in V$. 
Let $\mathcal{F}=(\mathcal{F}_0,\mathcal{F}_1,\mathcal{F}_2)$ where $\mathcal{F}_i \subseteq 2^V$ consists of all 
subsets $S \subseteq V$ of size $|S|=n/3$ for which $G$
contains a path $P$ such that
\begin{enumerate}
\item $P$ starts in $v_i$ and ends in $u_i$ such that $(u_i,v_{(i+1) \bmod 3}) \in E$;
\item $P$ goes through all the vertices in $S$ exactly once. 
\end{enumerate}
The family $\mathcal{F}$ can be computed in time $O^*(2^{H(1/3)n})=O^*(1.99^n)$ (where $H(x)=-x\log_2-(1-x)\log_2(1-x)$ is the binary entropy function) by dynamic programming. 

We now need to construct a system of families with infants for the resulting $(V,3,\mathcal{F})$-partition problem. 
We construct the required family for $p=\frac{n}{\Delta^2+1}$ and~$q=\Delta+1$.
Find greedily an independent set $I$ of size $p=\frac{n}{\Delta^2+1}$ in~$G^2$.
Assume that $I=\{1, \dots, p\}$ and let $R_i=N_G[i]$, $r_i=i$. Clearly, if $F \in \mathcal{F}_j$
contains an infant $r_i \in R_i$ then this infant is necessarily accompanied by one of its relatives since $F$ contains a Hamiltonian path.
\end{proof}

The corresponding algorithms allow also to solve weighted partition problems. In such problems, each subset $F$ of $\mathcal{F}_i$ is assigned a non-negative integer weight $w(F)$ and the goal is to find 
a partition of minimum total weight.

\begin{theorem}\label{thm:weighted}
If in Theorems~\ref{thm:simple} and \ref{thm:enc} one is given a weighted partition problem then the upper bounds on the running time and space are multiplied by $W$ where $W$ is the maximum weight of a subset.
\end{theorem}

Also, one can turn the algorithm to use polynomial space by providing an algorithm that enumerates the sets $\mathcal{F}_i$.

\begin{theorem}\label{thm:encenum}
Let $\mathcal{F}=(\mathcal{F}_1, \ldots, \mathcal{F}_k)$ and assume that there exists an algorithm that for any $i=1,\ldots, k$ enumerates the set $\mathcal{F}_i$ in time $T$ and polynomial space. Then one can turn 
algorithms from Theorems~\ref{thm:simple}, \ref{thm:enc}, \ref{thm:weighted} into polynomial space algorithms at the cost of multiplying the running time by~$T$.
\end{theorem}

As a corollary we get a polynomial-space algorithm for the case when each $\mathcal{F}_i$ is of polynomial size.

\begin{corollary}\label{cor:encpolyspace}
Let $\mathcal{R}$ and $(V,k,\mathcal{F})$ be as in Theorem~\ref{thm:encenum}.
If for all $i=1,\ldots,k$, $\mathcal{F}_i$ is enumerable in polynomial time (in particular, $|\mathcal{F}_i|=\poly(n)$) then the corresponding algorithm uses polynomial space.
\end{corollary}

%

We conclude the section by noting that the same bounds hold also
for the case when instead of partition one looks for a cover of $V$
by $k$ subsets from $\mathcal{F}_1, \ldots, \mathcal{F}_k$. We call the corresponding problem a \emph{$(V,k,\mathcal{F})$-covering problem}.

\begin{theorem}\label{thm:covering}
Theorems~\ref{thm:simple}, \ref{thm:enc}, \ref{thm:weighted}, \ref{thm:encenum} hold for $(V,k,\mathcal{F})$-covering problems.
\end{theorem}

\section{Properties of bounded degree graphs}\label{sec:bounded}
The following three lemmas are proved by Cygan and Pilipczuk~\cite{CP2013}. We slightly extend the statements and provide the proofs for the sake of completeness.
 
The lemma below allows to find in a graph a set of vertices of high degree with a better upper bound on its size than given by the standard averaging argument.

\begin{lemma}[\cite{CP2013}, Lemma~3.2]\label{lemma:coresize}
For any graph $G=(V,E)$ of average degree at most $d$, any integer $m\ge 1$ and any $\alpha \ge 1$ there exists 
$m\le D\le M$ such that $ |V_{>D}| \le \frac{nd}{\alpha D}$
where $M=\lfloor me^{\alpha+1}+1\rfloor$.
\end{lemma}
\begin{proof}\label{page:averageproofs}
Clearly,
\[\sum_{i=0}^{\infty}|V_{>i}|=\sum_{i=0}^{\infty}i|V_{=i}| \le nd \,.\]
If, on the other hand, $|V_{>i}|>\frac{nd}{\alpha i}$ for all $m\le i\le M$ then
\[\sum_{i=0}^{\infty}|V_{>i}| \ge \sum_{i=m}^{M}|V_{>i}|>\frac{nd}{\alpha}\sum_{i=m}^{M}\frac{1}{i}=
\frac{nd}{\alpha}\left(\sum_{i=1}^{M}\frac{1}{i}-\sum_{i=1}^{m-1}\frac{1}{i}\right) \ge
\frac{nd}{\alpha}\left(\ln M-\ln(em)\right) \ge 
nd \, ,\]
where the next to last inequality uses the well-known estimate for the harmonic series:
\[ \ln(i+1) \le 1+\frac{1}{2}+\ldots +\frac{1}{i} \le \ln i + 1=\ln(ie)\,. \]
\end{proof}

\begin{remark}
Such $D$ can be easily found in polynomial time by going through all the values $D=m,\ldots, M$.
\end{remark}

The next lemma shows that one can find 
a subset of vertices of linear size that is independent in the square of a graph.

\begin{lemma}[\cite{CP2013}, Lemma~3.1]\label{lemma:indsetsize}
For any graph $G=(V,E)$ of average degree at most $d$ and maximum degree at most $\Delta$
one can find in polynomial time a subset of vertices $B \subseteq V$ such that
\begin{enumerate}
\item the size of $B$ is linear: $|B| \ge \frac{n}{6\Delta d}$;
\item degrees of vertices from $B$ are small: for any $v \in B$, $\operatorname{deg}_G(v) \le 2d$;
\item $B$ is independent in $G^2$: for any $u \neq v \in B$, $N_G[u] \cap N_G[v] = \emptyset$.
\end{enumerate}
\end{lemma}
\begin{proof}
Clearly, $|V_{\le 2d}| \ge \frac{nd}{2d}=\frac{n}{2}$. The required set $B$ can be constructed by a straightforward greedy algorithm: while $V_{\le 2d}$ is not empty, take any $v \in V_{\le 2d}$, add it to $B$, and remove from $V_{\le 2d}$ the vertex $v$ together with all its neighbors in~$G^2$.
The number of such neighbors is at most $2d+2d(\Delta-1)=2d\Delta$. Hence 
at each iteration at most $2d\Delta+1$ vertices are removed and the total number of iterations is at least 
\[\frac{|V_{\le 2d}|}{2d\Delta+1} \ge \frac{n}{4d\Delta+2} \ge \frac{n}{6d\Delta} \,.\] 
\end{proof}


The following lemma allows us to find efficiently in a graph $G$ of average degree $d=O(1)$
a subset of vertices $Y$ of high degree such that $(G\setminus Y)^2$ contains a large enough 
independent set. The last inequality in the statement
can be seen as exponential savings in the running time.

\begin{lemma}[\cite{CP2013}, Lemma~3.4]\label{lemma:ugly}
For any constants $\nu \ge 1, \mu <1, a \ge 0, 0< c < 1$ there exists $\beta>0$ such that for any graph $G=(V,E)$ of average degree~$d=O(1)$ one can find in polynomial time subsets 
$A,Y \subseteq V$ such that:
\begin{enumerate}
\item $A \cap Y = \emptyset$;
\item $A$ is an independent set in $(G\setminus Y)^2$;
\item $2|Y|\le |A| \le cn$;
\item each vertex from $A$ has at most $2d$ neighbors in $G\setminus Y$:
$\forall v \in A,\, |\{u \in V \setminus Y \colon (u,v) \in E\}| \le 2d\,;$
\item 
\begin{equation}\label{eq:ugly1}
\binom{|A|}{|Y|}^a\nu^{|Y|}\mu^{|A|} < 2^{-\beta n}\,.
\end{equation}
\end{enumerate} 
\end{lemma}
\begin{proof}
Let $\alpha=\alpha(d, \nu, \mu, a,c)$ be a large enough constant to be defined later.
Using Lemma~\ref{lemma:coresize} we can find 
\begin{equation}\label{eq:cd}
 \frac{1}{12dc} \le D \le \frac{1}{12dc}{e^{\alpha+1}+1},
\end{equation}
such that $|V_{>D}| < \frac{nd}{\alpha D}$. Let $Y=V_{>D}$. Note that the graph $G\setminus Y$ has average degree at most $d$ and maximum degree at most~$D$. Lemma~\ref{lemma:indsetsize} allows us to find
a subset $A \subseteq V \setminus Y$ such that $A$ is independent in $(G \setminus Y)^2$,
for all $v \in A$, $\operatorname{deg}_{G \setminus Y}(v) \le 2d$ and 
\[|A| \ge \frac{n-|Y|}{6dD} \ge \frac{n}{12dD} \, ,\]
where the last inequality is true when $\alpha \ge 2d$, i.e., $\alpha$ is large enough.
Remove from $A$ arbitrary vertices such that $|A|=\frac{n}{12dD}$.
Because of~(\ref{eq:cd}), $\frac{n}{12dD} \le nc$.
To guarantee that $|A|>2|Y|$ it is enough to take $\alpha \ge 24d^2$.

We now show how to choose $\alpha$ such that the last inequality from the statement is satisfied. Using the well known estimates $\binom nk \le \left(\frac{en}{k}\right)^k$
and $\binom nk \le \binom{n}{k'}$ for $k \le k' \le \frac n2$ we get
\[ \binom{|A|}{|Y|}^a \le \left(\frac{en\alpha D}{12dDnd}\right)^\frac{nda}{\alpha D}=\left(\frac{e\alpha}{12d^2}\right)^\frac{nda}{\alpha D}=(\gamma \alpha)^\frac{nda}{\alpha D} \, ,\]
where $\gamma=\frac{e}{12d^2}$ is a constant. Thus we can upper-bound
(\ref{eq:ugly1}) as follows:
\begin{equation}\label{eq:ugly2}
\left(\left(\gamma \alpha\right)^\frac{da}{\alpha}\left(\nu^d \right)^\frac{1}{\alpha}\mu^\frac{1}{12d}\right)^\frac{n}{D} \, .
\end{equation}
Recall now that $\mu<1$ and note that $(\gamma\alpha)^\frac{da}{\alpha} \to 1$ and
$\left(\nu^d\right)^\frac{1}{\alpha} \to 1$ with $\alpha \to +\infty$. Thus for a large enough $\alpha$, (\ref{eq:ugly2}) is $\left(2^{-\beta'}\right)^\frac{n}{D}$ for a constant $\beta'>0$. It remains to recall that $D < e^\alpha$ and take $\beta=\beta'e^{-\alpha}$.
\end{proof}

\section{Solving Partition Problems}

For a subset $U \subseteq V$, let $b(U) \in \{0,1\}^n$ denote the characteristic
vector of the set $U$ (i.e., $b(U)[i]=1$ iff $i \in U$). 
In the analysis below we sometimes identify a bit vector $b(U)$ with a non-negative integer between $0$ and $2^n-1$ that it represents.

For a bit vector $b$, we denote the Hamming weight of $b$ by $w(b)$, i.e., the number of $1$'s in $b$. Note the following simple fact: for any two non-negative integers
$a$ and $b$,
\begin{equation}\label{eq:carries}
w(\operatorname{bin}(a))+w(\operatorname{bin}(b)) \ge w(\operatorname{bin}(a+b))
\end{equation}
and the equality holds iff there are no carries in $a+b$.

\subsection{Proof of Theorem~\ref{thm:simple}}
\begin{proof}[Proof of Theorem~\ref{thm:simple}]
Consider the following polynomials for $1 \le i \le k$:
\[P_i(x,y)=\sum_{F \in \mathcal{F}_i}x^{|F|}y^{b(F)} \, .\]
We claim that there is a solution to a $(V,\mathcal{F},k)$-partition problem iff
the monomial $x^ny^{b(V)}$ has a non-zero coefficient in $\prod_{i=1}^{k}P_i(x,y)$.
One direction of this statement is straightforward: if there exist $F_1\in \mathcal{F}_1, \dots, F_k \in \mathcal{F}_k$ such that $F_1 \sqcup \ldots \sqcup F_k=V$ then clearly
\[\prod_{i=1}^{k}x^{|F_i|}y^{b(F_i)}=x^ny^{b(V)} \, .\]
For the reverse direction, assume that the product of the polynomials contains
the monomial $x^ny^{b(V)}$. Because of the term $x^n$,
there exist $k$ subsets $F_1 \in \mathcal{F}_1, \dots, F_k \in \mathcal{F}_k$ such that 
$|F_1|+\ldots+|F_k|=n$. In other words, the total number of $1$'s in 
all characteristic vectors of $F_i$'s is exactly $n$. Moreover,
\begin{equation}\label{eq:sum}
b(F_1)+\ldots+b(F_k)=b(V)
\end{equation}
and the number of $1$'s in the characteristic vector of $V$ is also $n$.
Now,~(\ref{eq:carries}) implies that there are no carries in~(\ref{eq:sum}).
This in turn implies that $\{F_1, \ldots, F_k\}$ is a partition of~$V$.

Now it suffices to show that the coefficient of the monomial $x^ny^{b(V)}$ in $\prod_{i=1}^{k}P_i(x,y)$ can be found in time $O^*(2^n)$. 
Since $|F|\le n$ and $k\le n$, the degree of $x$ in $\prod_{i=1}^{k}P_i(x,y)$ does not exceed $n^2$. Therefore, in order to obtain univariate polynomials we can use Kronecker substitution~\cite{K1882}. Namely, we replace $y$ by $x^{n^2+1}$. Thus, for each $1 \le i \le k$ we consider a univariate polynomial $Q_i(x)$:
\[Q_i(x)=\sum_{F \in \mathcal{F}_i}x^{|F|}x^{(n^2+1)\cdot b(F)} \, .\]
It it easy to see that the coefficient of $x^{a_1+(n^2+1)\cdot a_2}$ (where $a_1\le n^2$) in $Q_i(x)$ equals the coeffient of $x^{a_1}y^{a_2}$ in $P_i(x)$, and vice versa.
In other words, we associate an integer from $[0..(n^2+1)2^n]$ with each $F\in\mathcal{F}_i$. This integer is an encoding of the set $F$ in $n+2\log{n}$ bits, s.t. the first $n$ bits indicate elements of $F$, the next $\log{n}$ bits are zeros, and the last $\log{n}$ bits are the binary expansion of $|F|$. 
We need to find the coefficient of $x^{n+(n^2+1)\cdot b(v)}$ in $Q=\prod_{i=1}^{k}Q_i(x)$, where $\deg{(Q)}=O^*(2^n), k\le n$. One can apply FFT $k-1$ times and get the desired upper bound on the running time.
\end{proof}

\begin{remark}\label{rem:counting}
Note that the coefficient of the monomial $x^ny^{b(V)}$ in the theorem equals the number of valid $k$-partitions. So this theorem (and all the related theorems) in fact count the number of valid partitions.
\end{remark}

\subsection{Proof of Theorem~\ref{thm:enc}}

\begin{definition}
For a matrix $M=(M[i,j])_{0\le i \le p-1, 0 \le j \le q-1} \in \mathbb{Z}_{\ge 0}^{p \times q}$ let
\begin{eqnarray*}
\operatorname{colweight}(M,j)&=&\sum_{i=0}^{p-1}M[i,j]\,,\quad
\operatorname{weight}(M)=\sum_{i=0}^{p-1}\sum_{j=0}^{q-1}M[i,j]=\sum_{j=0}^{q-1}\operatorname{colweight}(M,j)\,,\\
\operatorname{rowcode}(M,i)&=& -M[i,0]+\sum_{j=1}^{q-1}2^j\cdot M[i,j]\,,\quad
\operatorname{rowsum}(M)=\sum_{i=0}^{p-1}\operatorname{rowcode}(M,i)\,,\\
\operatorname{code}(M)&=&\sum_{i=0}^{p-1}(2^q - 1)^i\cdot \operatorname{rowcode}(M,i)\,.
\end{eqnarray*}
\end{definition}

\begin{definition}
A matrix $M \in \mathbb{Z}_{\ge 0}^{p \times q}$ is called \emph{row-normalized} if $\operatorname{rowcode}(M,i) \ge 0$ for all $0 \le i \le p-1$.
\end{definition}
\begin{remark}
In the analysis below we will need the following simple estimates. Let $E \in \{0,1\}^{p \times q}$. Then
\begin{enumerate}
\item $\operatorname{rowcode}(E,i)<0$ iff $E[i]=[1,0,0,\ldots,0],$
\item $\operatorname{rowcode}(E,i) \le 2^q-2 $.
\item $\operatorname{code}(E) \le (2^q-2)\cdot \sum_{i=0}^{p-1}(2^q-1)^i<(2^q-1)^p$ (assuming $q \ge 2$).
\end{enumerate}
\end{remark}

\begin{lemma}
\label{lemma:system}
The expansion of $X\in\mathbb{Z}_{\ge 0}$ in powers of $b>1$ as $X=\sum_{i=0}^{\infty}x_i\cdot b^i, x_i\ge0$ has the minimal value of the sum of digits $\sum_{i=0}^{\infty}x_i$ iff $\forall i: 0\le x_i<b$ (i.e., $X$ is written in the numeral system of base $b$).
\end{lemma}
\begin{proof}
If $x_i\ge b$, then we can reduce $x_i$ by $b$ and increase $x_{i+1}$ by $1$, decreasing the sum of digits by $b-1>0$. The other direction follows from the uniqueness of representation in base~$b$.
\end{proof}

\begin{lemma}
\label{lemma:coding}
Let $q\ge 2$ and $E \in \{0,1\}^{p \times q}$
and $M \in \mathbb{Z}_{\ge 0}^{p \times q}$ be row-normalized matrices.
If
$\operatorname{colweight}(M,0)=\operatorname{colweight}(E,0)$, 
$\operatorname{weight}(M)=\operatorname{weight}(E)$,
$\operatorname{rowsum}(M)=\operatorname{rowsum}(E)$,
$\operatorname{code}(M)=\operatorname{code}(E)$,
then $M=E$.
\end{lemma}

\begin{proof}
The claim follows from Lemma~\ref{lemma:system}.
Since for all $i$, $\operatorname{rowcode}(E,i) \le 2^q-2 $, $\operatorname{code}(E)$ has the minimal sum of digits in base $(2^q-1)$ system. This in turn implies that for each $i$, $\operatorname{rowcode} (E,i)=\operatorname{rowcode}(M,i)$. Then the first columns of matrices $E$ and $M$ are equal modulo $2$, because parities of rowcodes depend only on the first column. Since $\operatorname{colweight}(M,0)=\operatorname{colweight}(E,0)$ we conclude that the first columns of $M$ and $E$ are equal.

Now each $\operatorname{rowcode}(E,i)$ has the minimal sum of digits in the system of base $2$, which means that $\operatorname{weight}(E)$ has the minimal possible value for these $\operatorname{rowcodes}$. It follows from Lemma~\ref{lemma:system} that each $M[i,j]$ must be equal to $E[i,j]$.
\end{proof}

\begin{definition} 
An injective function $\alpha \colon U \to \{0,\dots,p-1\} \times \{0,\dots, q-1\}$ is called a \emph{matrix representation} of a set $U$. For such $\alpha$ and $S \subseteq U$, a~\emph{characteristic matrix} $M_\alpha(S) \in \{0,1\}^{p \times q}$ is defined as follows: $i \in U$ iff $M_\alpha(S)[\alpha(i)]=1$.
\end{definition}

\begin{proof}[Proof of Theorem~\ref{thm:enc}]
Let $\mathcal{R}=((R_1,r_1), \dots, (R_p,r_p))$ be a {$(p,q)$-system of families with infants}
for a $(V,k,\mathcal{F})$-partition problem. Append arbitrary elements from $V$ to families so that the size of each family equals $q$ and the families are still disjoint (this is possible since $pq \le n$). Denote the union of  families by $R$ and the rest of $V$ by~$L$. 
For each family~$R_i$, fix an order of its elements such that
the $0$th element is~$r_i$. Now consider a matrix representation $\alpha \colon V \to \{0,\dots,p-1\} \times \{0,\dots, q-1\}$ defined as follows.  If $v$ is the $j$th element of $R_i$, then $\alpha(v)=(i,j)$.
We encode each set $F\in\mathcal{F}_i$ by parts. We encode vertices from $F\cap L$ using the standard technique from Theorem~\ref{thm:simple}. To encode vertices from $F\cap R$ we use the characteristic matrix $M_\alpha(F\cap R)$. Note that $M_\alpha(F\cap R)$ is a {row-normalized matrix}, because if $F$ contains an infant $r_i$ of a family~$R_i$, then it must contain at least one other element from the same row.
Consider the following polynomials for $1 \le i \le k$:
\[P_i(x,y,z,s,t,u)=\sum_{F \in \mathcal{F}_i}x^{|F\cap L|}\cdot y^{b(F\cap L)} \cdot z^{\operatorname{colweight}(M,0)} \cdot s^{\operatorname{weight}(M)} \cdot t^{\operatorname{rowsum}(M)} \cdot u^{\operatorname{code}(M)}\, ,\]
where $M=M_\alpha(F\cap R)$ and $b(F\cap L)$ is an integer from $0$ to $2^{|L|}-1$.

There is a solution to a $(V,\mathcal{F},k)$-partition problem iff
the monomial \[x^{|L|}y^{b(L)}z^{\operatorname{colweight}(R,0)} s^{\operatorname{weight}(R)} t^{\operatorname{rowsum}(R)} u^{\operatorname{code}(R)}\] has a non-zero coefficient in $\prod_{i=1}^{k}P_i(x,y,z,s,t,u)$. Indeed, as it was shown in Theorem~\ref{thm:simple}, $x^{|L|}y^{b(L)}$ corresponds to partitions of $L$. Lemma~\ref{lemma:coding} implies that only partitions of $R$ may have the term $z^{\operatorname{colweight}(R,0)} s^{\operatorname{weight}(R)} t^{\operatorname{rowsum}(R)} u^{\operatorname{code}(R)}$. Note that the degrees of $x,z,s$ in $\prod_{i=1}^{k}P_i(x,y,z,s,t,u)$ are bounded from above by $n^2$, the degree of $y$~--- by $n\cdot2^{|L|}$, the degree of $t$~--- by $n^2\cdot2^q$, the degree of $u$~--- by $n\cdot(2^q-1)^p$. Now we can apply Kronecker substitution: $y=x^{(n+1)^2}, z=x^{(n+1)^32^{|L|}}, s=x^{(n+1)^52^{|L|}}, t=x^{(n+1)^72^{|L|}}, u=x^{(n+1)^92^{|L|}2^q}$. The running time of FFT is bounded by the degree of the resulting univariate polynomial, i.e.
\[O^*(\poly(n)2^{|L|}2^q(2^q-1)^p)=O^*(2^{n-pq}(2^q-1)^p2^q)=O^*\left(2^n\cdot\left(\frac{2^q-1}{2^q}\right)^p \cdot 2^q\right)\,.\]
\end{proof}

\subsection{Proofs of Theorems~\ref{thm:weighted}, \ref{thm:encenum}, \ref{thm:covering}}

\begin{proof}[Proof of Theorem~\ref{thm:weighted}]
Following the proofs of Theorem~\ref{thm:simple} and \ref{thm:enc}, we introduce polynomials corresponding to $\mathcal{F}_i$'s. But in the weighted partition problem we multiply each monomial by $z^w$, where $z$ is a new variable and $w$ is the weight of the corresponding set. For example, in Theorem~\ref{thm:simple} the new polynomials would look as follows:
\[P_i(x,y,z)=\sum_{F \in \mathcal{F}_i}x^{|F|}y^{b(F)}z^{w(F)} \, ,\]
where $w(F)$ is the weight of $F$. Now it is clear that there exists a partition of total weight $w$ iff
the monomial $x^ny^{b(V)}z^w$ has a non-zero coefficient in $\prod_{i=1}^{k}P_i(x,y,z)$. We apply Kronecker substitution $y=x^{n^2+1}, z=x^{(n^2+1)n2^n}$ and use FFT to find all the coefficients of the product. Now we just need to find the smallest $w$, s.t. the coefficient of $x^{n+(n^2+1)\cdot b(V)+(n^2+1)n2^nw}$ does not equal $0$. Since the degree of the polynomial is at most $O^*(2^nW)$, the running time of the algorithm is $O^*(2^nW)$.
\end{proof}

\begin{proof}[Proof of Theorem~\ref{thm:encenum}]
Assume that we need to find the coefficient of the monomial $x^m$ in $P=\prod_{i=1}^{k}P_i(x)$, where $\deg(\Pi)\le d$. 
The theorem assumption claims that we can evaluate each $P_i$ at any point in time $O^*(T)$. Note that in order to get one specific coefficient of $P(x)$, we just need to list all the Fourier coefficients of $P$. Indeed, the coefficient of $x^m$ in $P$ equals 
$$\frac{1}{d}\sum_{i=0}^{d-1}\omega_d^{-im}P(\omega_d^k)=\frac{1}{d}\sum_{i=0}^{d-1}\omega_d^{-im}\prod_{i=1}^{k}P_i(\omega_d^k).$$
Since each $P_i(x)$ can be evaluated in time $O^*(T)$, we need $O^*(d\cdot T)$ steps and only polynomial space to find one coefficient of the product.
%
\end{proof}

\begin{proof}[Proof of Theorem~\ref{thm:covering}]
The proofs are very similar to the proofs for partition problems, the only difference is that now we consider polynomials
\[P'_i(x,y)=\sum_{F \in \mathcal{F}_i} \prod_{v\in F}({1+x\cdot y^{b(\{v\})}}) \, .\]
Let $\mathcal{F}'_i=\bigcup_{F\in\mathcal{F}_i}{2^F}$. Clearly, a $(V,\mathcal{F},k)$-covering problem has a solution iff a $(V,\mathcal{F}',k)$-partition problem has one. The polynomial $P'_i(x,y)$ corresponds to the polynomial $P_i(x,y)$ for $\mathcal{F}'$ from the proof of Theorem~\ref{thm:enc}. Namely, $x^my^{b(U)}$ has a non-zero coefficient in $P'(x,y)$ iff $U\subseteq\mathcal{F}_i$ and $|U|=m$. Again, the degree of $x$ in $\prod_{i=1}^{k}P'_i(x,y)$ is less than $n^2+1$, so we can apply Kronecker substitution $y=x^{n^2+1}$. Thus, there exists a valid cover iff the coefficient of the monomial $x^{n+(n^2+1)\cdot b(v)}$ does not equal $0$. From now on we can follow the proofs of Theorems~\ref{thm:simple}, \ref{thm:enc}, \ref{thm:weighted}, \ref{thm:encenum}.
\end{proof}

\section{The chromatic number problem}

\begin{definition}
In the \emph{list coloring problem} each vertex $v$ of input graph is assigned a list $L_v$ of allowed colors and the task is to properly color a graph such that each vertex
is given a color from its list.
\end{definition}

To reduce the search space in the list coloring problem we introduce the following problem.

\begin{definition}
In the \emph{coloring with preferences problem} each vertex $v$ together with a list $L_v$ is given a \emph{preferred color} $p_v \in L_v$ and the task is to color the graph properly such that each vertex $v$ is given a color from $L_v$ and moreover at 
least one vertex from $N[v]$ is given the color~$p_v$.
\end{definition}

\begin{lemma}
Let $G$ be a graph, $\{L_v\}_{v \in V}$ be a set of list colors for its vertices, and
$\{p_v \in L_v\}_{v \in V}$ be a set of preferred colors. Then 
there is a solution for an instance $(G, \{L_v\})$ of the list coloring problem iff
there is a solution for an instance $(G, \{L_v\}, \{p_v\})$ of the coloring with preferences problem.
\end{lemma}
\begin{proof}
Obviously, if there is a coloring satisfying preferences then it is also a list-coloring.
For the reverse direction, consider a proper list coloring such that for some vertex 
$v$ neither of the vertices from $N[v]$ is given the color $p_v$. One can then change the color of $v$ to $p_v$. This clearly does not violate any coloring constraints and moreover it strictly increases the number of vertices that are given their preferred color.
Thus, after repeating this operation at most $n$ times one arrives at a valid coloring
with preferences.
\end{proof}

Checking whether a graph $G$ has a proper $k$-coloring is a $(V,\mathcal{F},k)$-partition problem where for all $i=1,\ldots,k$, $\mathcal{F}_i=IS(G)$, the set of all independent sets of $G$. Note that Theorem~\ref{thm:simple} already implies 
a $O^*(2^n)$ time and space algorithm for the Chromatic Number problem. An algorithm with the same time and space bounds was given recently by Bj\"{o}rklund et al.~\cite{BHK2009}.

For $k$-coloring with preferences the families $\mathcal{F}_i$ are defined differently. Namely, $\mathcal{F}_i$
consists of all independent sets $I$ of $G$ that can be assigned the color $i$
without violating any list constraints and preferred color constraints:
\begin{itemize}
\item (list constraints): $\forall v \in I$, $ i \in L_v$;
\item (preferred color constraints): $\forall v \in V$  such that $p_v=i$, 
$N_G[v] \cap I \neq \emptyset$.
\end{itemize}
Using this interpretation of the coloring with preferences problem,
we give an algorithm solving the Chromatic Number problem
on graphs of bounded average degree in time $O^*((2-\varepsilon)^n)$ and space.

\begin{theorem}
There is an algorithm checking whether for a given graph $G$ of average degree $d=O(1)$ there exists a proper $k$-coloring in time $O^*((2-\varepsilon(d))^n)$
and exponential space.
\end{theorem}
\begin{proof}
First consider the case $k \ge 2d$. Note that $|V_{\ge k}| \le \frac{nd}{k} \le n/2$.
Then one can find a proper $k$-coloring 
of a graph $G[V_{\ge k}]$ in time $O^*(2^{n/2})$. Such a coloring can be easily extended to the whole graph (since there always exists a vacant color for a vertex of degree at most $k-1$). Thus, in the following we assume that $k < 2d=O(1)$.

Let $\nu \ge 1, \mu<1,a \ge 0, 0<c<1$ be constants to be defined later and let $A,Y \subseteq V$ be as provided by Lemma~\ref{lemma:ugly}. For $Y$ we try all possible colorings in time $k^{|Y|}$. A~fixed coloring of $Y$ produces a list coloring problem for~$G \setminus Y$.
%
%
Let $L$ be one of the most frequent color lists of vertices from~$A$ and assume w.l.o.g. that $L=\{1, \dots, l\}$ (hence $l \le k$). Let 
$C=\{v \in A \colon L_v = L\}$. Since there are at most $2^k$ different lists,
$|C| \ge |A|/2^k$. Remove arbitrary vertices from $C$ so that $|C|=\frac{|A|}{2^k}$.

Now, first remove a few vertices
from $C$ such that $|C|$ is divisible by 
$l+1$ and partition $C$ arbitrarily into $l+1$ subsets $C_1, \dots, C_{l+1}$
of size $\frac{|C|}{l+1}$.
For all $i=1, \ldots,l$, assign the preferred color $p_v=i$ for all vertices $v$ from $C_i$
and assign preferred colors arbitrarily for all the remaining vertices ($V\setminus Y \setminus \cup_{i=1}^{l}C_i$).

We are ready to construct sets with infants for $V\setminus Y$.
For $j=1, \dots, \frac{|C|}{(l+1)}$, let
\begin{equation}\label{eq:ri}
R_j = \bigcup_{i=1}^{l+1}N_G[C_i[j]] \text{ and } r_j=C_{l+1}[j] \, .
\end{equation}
See Figure~\ref{fig:ris} for a visual explanation.
First of all note that all $R_j$'s are disjoint since they consist of
the closed neighbourhoods of vertices from $A$ and according to Lemma~\ref{lemma:ugly} such neighbourhoods are disjoint. To show that the sets $R_j$'s are indeed sets with infants,
consider a set $F \in \mathcal{F}_i$ and assume that $r_j \in F$. Recall that $\mathcal{F}_i$ consists of the sets that can be assigned the color~$i$.
Our goal is to show that $F$ also contains at least one other element of~$R_j$. Although we do not know the color of $r_j$, we do know that its color is from $L$. Note that each family must have at list one element of each color, hence $r_j$ has a relative in~$F$. Below we formally prove that the constructed system satisfies all properties of families with infants.

\begin{figure}
\begin{center}
\begin{tikzpicture}

\draw (0,0) rectangle (8,4);
\foreach \y/\i in {3.75/C_1, 3.25/C_2, 2.5/\vdots, 1.75/C_i, 1/\vdots, 0.25/C_{l+1}}
  \node at (4,\y) {$\i$};
\foreach \y in {0.5, 1.5, 2, 3, 3.5}
  \draw (0,\y) -- (8,\y);
\foreach \y/\i in {3.75/1, 3.25/2, 1.75/i, 0.25/l+1}
  \node[anchor=east] at (-0.1,\y) {$\i$};
\foreach \x/\i in {0.25/1, 0.75/2, 6.25/j, 7.75/\frac{|C|}{l+1}}
  \node[anchor=north] at (\x,-0.2) {$\i$};  

\node at (0.5,4.5) {$C$};

\node[rectangle,anchor=west,text width=50mm] at (8.2,1.75) {the preferred color of all vertices from $C_i$ is $i$\\ (for all $1 \le i \le l$)};

\draw (6,-0.1) rectangle (6.5,4.1);

\node[anchor=south west,rectangle,text width=80mm] at (6.25,4.2) {$R_j = \bigcup_{i=1}^{l+1}N_G[C_i[j]] \text{ and } r_j=C_{l+1}[j]$};

\end{tikzpicture}
\end{center}
\caption{Schematic explanation for $R_i$'s.}
\label{fig:ris}
\end{figure}
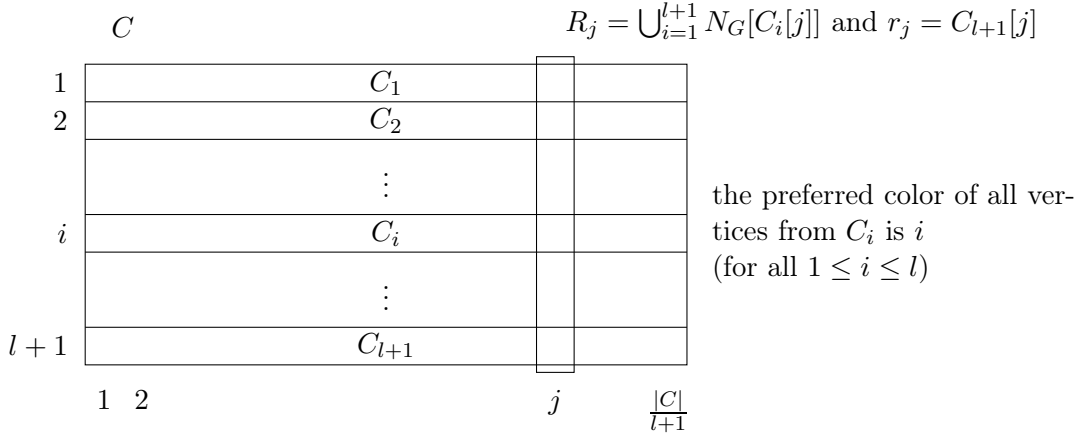

The fact 
that $r_j$ belongs to $F \in \mathcal{F}_i$ means that the vertex $r_j$ can be assigned the color~$i$.
In other words, $i \in L$. Recall from the definition (\ref{eq:ri}) of $R_j$ that
it contains, in particular, $N_G[C_i[j]]$. Since the preferred color of the vertex $C_i[j]$ is~$i$, at least one of $N_G[C_i[j]]$ must be given the color~$i$. Thus, besides of $r_j$, the set $F$ must contain also at least one element from $N_G[C_i[j]] \subseteq R_j$. 
By choosing a small enough value for $c$ one can guarantee that $pq \le n$.
We conclude that $((R_1, r_1), \dots, (R_p,r_p))$ is
a $(p,q)$-family of sets with infants for $p=\frac{|C|}{l+1}=\frac{|A|}{2^k(l+1)}$ and $q=(2d+1)(l+1) \le (2d+1)(k+1) \le (2d+1)^2$.

Theorem~\ref{thm:enc} implies that the running time of the resulting algorithm is at most
\[k^{|Y|}\cdot \left( 2^{|V \setminus Y|} \cdot 
\left(\frac{2^q-1}{2^q}\right)^p\right) \le 2^n \cdot \left( k^{|Y|} \cdot 
\left(\frac{2^q-1}{2^q}\right)^p\right) \, .\]

We now choose the constants $\nu,\mu,a$ so that (\ref{eq:ugly1}) implies that the expression in parentheses is at most $2^{-\beta n}$ for a constant $\beta > 0$. Let $a=0$, $\nu=k$, and 
\[\mu=\left(\frac{2^q-1}{2^q}\right)^{\frac{1}{2^k(l+1)}}\]
(recall that $k<2d=O(1)$ and $l \le k=O(1)$).
It is easy to see that (\ref{eq:ugly1}) then implies that the total running time is $2^{(1-\beta)n}$ for a constant $\beta>0$.


\end{proof}

\begin{remark}
It is not difficult to see that the presented algorithm actually solves a more general problem known as maximum $k$-cut. In this problem the goal is to partition the vertices into $k$ parts such that the number of edges joining different parts is maximal possible.
\end{remark}

\section{The traveling salesman problem}

%
%
%

\begin{theorem}\label{thm:tsp}
The traveling salesman problem on graphs of average degree $d=O(1)$ can be solved 
in time $O^*(W \cdot (2-\varepsilon(d))^n)$ and polynomial space $O(n^2\log W)$ ($W$ being the maximal edge weight).
\end{theorem}
\begin{proof}
Let $k$ be a parameter to be defined later. Guess $k$ vertices $v_1, \ldots, v_{k} \in V$ that split an optimal Hamiltonian cycle into $t$ paths of length $n/k$. All such $k$-tuples
can be enumerated in time~$n^k$.
Then the corresponding weighted $(V,k,\mathcal{F})$-partition problem is defined as follows.
$\mathcal{F}=(\mathcal{F}_1, \ldots, \mathcal{F}_k)$ where $\mathcal{F}_i$ consists
of all subsets $S \subseteq V$ of size $|S|=\frac{n}{k}$ for which $G$ contains a Hamiltonian path $P$ such that
$P$ starts in $v_i$, goes through all vertices from $S$, and ends in $u_i$ such that $(u_i,v_{(i \bmod k) +1}) \in E$. The weight $w(S)$ of such a set $S$ is equal to the minimal possible weight of a path $P$ (including the weight of the edge $(u_i,v_{(i \bmod k) +1})$). It is not difficult to see that solving the traveling salesman problem is equivalent to solving the $(V,k,\mathcal{F})$-partition problem if the vertices $v_1,\dots,v_k$ are guessed correctly.

Note that any $\mathcal{F}_i$ can be enumerated in time
$\binom{n}{\frac{n}{k}} \cdot \left(\frac{n}{k}\right)! $
(the first term is for guessing the subset of vertices $S$, the second one
is for guessing the order of these vertices in an optimal path~$P$).
Recall also that guessing the vertices $v_1, \ldots, v_k$ requires $O(n^k)$ time. By choosing $k=\sqrt{n}$ we turn both these estimates into subexponential $2^{o(n)}$.

Now we turn to constructing the corresponding system of families with infants. Let $\nu \ge 1, \mu<1,a \ge 0, 0<c<1$ be constants to be defined later. Let $A,Y \subseteq V$ be provided by Lemma~\ref{lemma:ugly}. 
Consider an optimal Hamiltonian cycle $C$ in the graph. Let $Y' \subseteq V$ be the successors of the vertices from $Y$ in the cycle $C$ and let $A'=A \setminus Y'\setminus\{v_1,\ldots,v_k\}$. 
Note that $|A'| \ge \frac{|A|}{2}-k$ (since $|A| \ge 2|Y|$) and for each vertex $v \in A'$ its predecessor $u$ in the cycle $C$ belongs to $V \setminus Y$.

Let $A'=\{1, \ldots, p\}$. Then for all $i=1,\ldots, p$, $R_i=N_{G\setminus Y}[i]$ and $r_i=i$. Clearly, $|R_i| \le q=2d+1$. By choosing $c<\frac{1}{2d+1}$ we can guarantee that $pq \le n$. All $R_i$'s are disjoint since $A'$ is an independent set in $(G \setminus Y)^2$. Finally, if the set $Y'$ is guessed correctly (i.e., $Y'$ are indeed successors of $Y$ in the optimal cycle~$C$) then $((R_1,r_1), \dots, (R_p,r_p))$ is a $(p,q)$-system of families with infants. Indeed, if $r_i \in F$ for some $F \in \mathcal{F}_j$, then $r_i$'s predecessor in $C$ must lie in $V\setminus Y$, i.e., in $R_i$
(and $r_i$ must have a predecessor since $r_i \neq v_1, \ldots, v_k$). 

By Theorems~\ref{thm:weighted} and \ref{thm:encenum} the total running time does not exceed
\[2^{o(n)} \cdot \binom{|A|}{|Y|} \cdot 2^n \cdot \left(\frac{2^q-1}{2^q}\right)^p\cdot W.\]

Recall that $p=|A'|-k \ge |A|/2-|A|/6=|A|/3$ for large enough~$n$ (since $|A|=cn$ and $k=\sqrt{n}$). Choose $a=1$, $\nu=1$, $\mu=\left( \frac{2^{2d+1}-1}{2^{2d+1}}\right)^{1/3}$. Then (\ref{eq:ugly1}) implies that 
\[\binom{|A|}{|Y|}\cdot \left(\frac{2^q-1}{2^q}\right)^p < 2^{-\beta n}\]
for a constant $\beta > 0$. Thus the total running time is $O^*(W\cdot (2-\varepsilon)^n)$.

\end{proof}

\section{Counting perfect matchings}
The algorithm for counting perfect matchings uses many ideas from the algorithm for the traveling salesman problem presented in Theorem~\ref{thm:tsp}.

\begin{lemma}\label{lm:antimatching}
Let $G$ be a graph of average degree $d=O(1)$. Then $\bar{G}$ (the complement of~$G$) contains a matching consisting of $\frac{n}{2}-O(1)$
edges.
\end{lemma}
\begin{proof}
Clearly $|V_{\ge \frac{n}{3}}| \le \frac{nd}{n/3}=3d$. After removing all these vertices from $G$ we get a graph with at least $n-3d$ vertices such that the degree of each vertex is at most $\frac{n}{3}$. This implies that the degree of each vertex in the complement of this graph is at least $n-3d-1-\frac{n}{3}$. This is at least $\frac{n}{2}$ for large enough $n$. By Dirac's theorem~\cite{D1952} this graph is Hamiltonian and hence contains a perfect matching.
\end{proof}

\begin{theorem}
The number of perfect matchings in a graph $G$ with $2n$ vertices of average degree $d=O(1)$ can be found in time $O^*((2-\varepsilon(d))^n)$ and polynomial space.
\end{theorem}
\begin{proof}
Assume that the vertices of $V=\{1, \ldots,2n\}$ are numbered in such a way
that 
\begin{equation}\label{eq:loops}
(1,n+1),(2,n+2), \ldots, (m,n+m) \not \in E
\end{equation}
where $m=n-O(1)$. Such a numbering exists due to Lemma~\ref{lm:antimatching} (and can be found efficiently since we can find a maximum matching in $\bar{G}$ in polynomial time). Following~\cite{B2012} and \cite{CP2013} we reduce the problem of counting perfect matchings to a problem of counting cycle covers of a special type. We construct an auxiliary multigraph $G'=(V',E')$ where $V'=\{1,\ldots,n\}$
and each edge $(i,j) \in E$ is transformed into an edge 
$e=((i \mod n)+1, (j \mod n)+1) \in E'$
with the label $l(e)=\{i,j\}$. In other words, we contract each pair of vertices
$(1, n+1), \ldots, (n,2n)$ and on each edge we keep a label showing where it originates from. 
Any two vertices in $G'$ are joined by at most $4$ edges.
The average degree of $G'$ is at most~$2d$. 

Recall that a cycle cover of a multigraph is a collection of cycles such that each vertex belongs to exactly one cycle. In other words, this is a subset of edges such that each vertex is adjacent to exactly two of these edges (and a self-loop is thought to be adjacent to its vertex two times).

The important property of the graph $G'$ is the following:
each perfect matching in $G$ corresponds to a cycle cover $C \subseteq E'$ in $G'$
such that $\cup_{e \in C}l(e)=V$ and vice versa. Indeed, each vertex $i$ in $G'$ is adjacent to exactly two edges. These two edges have different labels so they correspond to edges in the original graph $G$ that match both $i$ and $i+n$. 

We have reduced the problem to counting cycle covers with disjoint labels in~$G'$.
We further reduce the problem to counting cycle covers without self-loops.
Note that by (\ref{eq:loops}) $G'$ has at most $s=O(1)$ self-loops. For each such loop $e=(i,i)$ we can consider two cases: to count the number of cycle covers with $e$ we 
count the number of cycle covers in $G'$ without the vertex~$i$; to count the number 
of cycle covers without $e$ we can just remove the loop $e$ from $G'$ and count the number
of cycle covers. This way, we reduce the problem to $2^s=O(1)$ problems of counting cycle
covers in a multigraph without self-loops.

Note that we only need to ensure that the labels on adjacent edges are disjoint.
Namely, one of the edges adjacent to a vertex $i$ in $G'$ must contain $i$ in its label while the other one must contain $i+n$.

Recall that Remark~\ref{rem:counting} claims that we can use Theorem~\ref{thm:enc} to find the number of valid partitions. We consider two important special cases. 
\begin{enumerate}
\item 
{\em Counting the number of covers by $k$ cycles of length at most $\sqrt{n}$.}
The good thing with covers by short cycles is that all such cycles can be easily enumerated in subexponential time and polynomial space. Indeed, to enumerate all possible short cycles one first goes through all possible length values $l$ (at most $\sqrt n$ choices), then through all subsets of vertices of a cycle (at most $\binom nl \le n^{\sqrt{n}}=2^{o(n)}$ choices), then through all orders of the vertices (at most $l! \le \sqrt{n}! \le \sqrt{n}^{\sqrt{n}}=2^{o(n)}$ choices), and for each 
pair of consecutive vertices go through all possible
edges that join them (at most $4^{\sqrt{n}}=2^{o(n)}$ choices). To guarantee that each cycle is counted only once instead of considering all $l!$ orderings we consider $\frac{l!}{2l}$
of them.

We then proceed as in Theorem~\ref{thm:tsp}. Namely, let $A,Y \subseteq V'$ be provided by Lemma~\ref{lemma:ugly}. We then go through all possible sets $Y' \subseteq A$ such that $|Y'| \le |Y|$. The set $Y'$ is thought as predecessors of the vertices from $Y$ in cycle covers that lie in $A$ (since we consider cycle covers without self-loops each vertex has a predecessor). Then, each vertex from $A'=A \setminus Y'$ has a predecessor in $V' \setminus Y'$ and we can find a system of families with infants in exactly the same way as in Theorem~\ref{thm:tsp}. Note that for each $Y'$ we enumerate only short cycles where $Y'$ is indeed the set of predecessors of~$Y$. 
Thus, the number of covers by cycles with disjoint labels of length at most $\sqrt{n}$ can be counted in time $O^*((2-\varepsilon(d))^n)$ and polynomial space. 

Note that in fact we count ordered cycle covers
(where the order of cycles does matter), so we should also divide the resulting number
by~$k!$.

\item 
{\em Counting the number of covers by $k$ cycles of length greater than $\sqrt{n}$.}
Long cycles cannot be quickly enumerated, but we know for sure that if all the cycles
have length at least $\sqrt{n}$ then the number $k$ of cycles is at most $\sqrt{n}$. 
We go through all possible cases of lengths of $k$ cycles
(at most $n^{\sqrt n}=2^{o(n)}$ choices). Then we proceed as in Theorem~\ref{thm:tsp}. 
Namely, 
for each cycle of length $l$
we go through all possible tuples of $t=\frac{l}{\sqrt{n}}$ vertices $v_1, \ldots, v_t$ (call them pivot vertices)
such that:
\begin{itemize}
\item $v_1$ is the minimal vertex of the considered cycle (this ensures that the corresponding cycle is counted only once);
\item the length of a subpath of the considered cycle between $v_i$ and $v_{i+1}$
is $\sqrt{n}$ for all $i=1, \ldots, t-1$
\end{itemize}
(at most 
$n^{2\sqrt{n}}=2^{o(n)}$ choices since for long cycles there are at most 
$2\sqrt{n}$ such vertices in total). 
For each vertex from the tuples we also go through all possible labels of two edges adjacent to it (at most $(2n)^{2\sqrt{n}}$ choices since each vertex of $G'$ has at most $2n$ adjacent edges).
The subpaths of length at most $\sqrt{n}$ can be enumerated in subexponential time and polynomial space. Then we can solve the corresponding partition problem 
in time $O^*((2-\varepsilon(d))^n)$ and polynomial space as shown in Theorem~\ref{thm:tsp}.

As in the previous case, we divide the resulting number by $k!2^k$ since we counted ordered cycle covers and each cycle was counted twice (in clockwise and anticlockwise directions).
\end{enumerate}

We proceed to the general case. To count the total number of cycle covers we first go through 
all $O(n^2)$ pairs $(k_1,k_2)$ where $k_1,k_2$ is the number of short and long cycles, respectively. We then go through all possible values of lengths of all long cycles
and all tuples of pivot vertices. Given the lengths of long cycles, by $k'_2$ we denote the total number of subpaths of length $\sqrt{n}$ which are needed to cover the long cycles. When these objects are fixed we take subsets $A,Y \subseteq V'$ as given by Lemma~\ref{lemma:ugly}. We then go through all possible $Y' \subseteq A$ such that $|Y'| \le |Y|$ and $Y'$ are thought as predecessors of $Y$ in $A$. 

We define the corresponding $(V,k_1+k'_2,\mathcal F)$-partition problem as follows. For $1\le i\le k_1$, ${\mathcal F}_i$ contains all cycles of length $\le\sqrt{n}$ where adjacent edges have disjoint labels. For $k_1+1\le i\le k_1+k'_2$, ${\mathcal F}_i$ is the set of subpaths of length $\sqrt{n}$ (note that some of  ${\mathcal F}_i$'s correspond
to the last part of a long cycle and hence consist of subpaths of length
smaller than~$\sqrt{n}$) where adjacent edges have disjoint labels. In order to get the number of unordered cycle covers, we need to divide the resulting value by $k_1!k_2!$.

Let then $A'$ be $A \setminus Y'$ with all pivot vertices removed. Using $A'$ we can find a system of families with infants exactly as it is done in Theorem~\ref{thm:tsp}.

\end{proof}

\section{Further directions}
As a conclusion we would like to pose some open questions which arise from the suggested algorithms. The first question is to obtain an efficient polynomial-space algorithm for the chromatic number problem. For example, $O^*(2^n)$-time algorithm for coloring in general graphs or $O^*((2-\varepsilon)^n)$-algorithm for coloring in bounded maximum/average degree graphs. A long-standing open question is to solve the traveling salesman problem in time $O^*((2-\varepsilon)^n)$. Any evidence of hardness of TSP would be of great interest as well.

\bibliographystyle{plain}
\bibliography{allrefs}


\end{document}